\documentclass[envcountsame,runningheads]{llncs}
\usepackage[utf8]{inputenc}
\usepackage{amsmath}
\usepackage{graphicx}
\usepackage{algpseudocode}
\usepackage{algorithm}
\usepackage{enumerate} 
\usepackage{setspace}
\usepackage{multirow}
\usepackage{graphics}
\usepackage{amssymb}
\usepackage{authblk}
\usepackage{array}
\usepackage{color} 
\usepackage{url}

\newif\ifjVer
\jVertrue
\newif\ifC
\Cfalse

\algnewcommand\algorithmicswitch{\textbf{switch}}
\algnewcommand\algorithmiccase{\textbf{case}}

\def\s#1{\mbox{\boldmath $#1$}}
\def\itbf#1{\textit{\textbf{#1}}}
\def\RLE{{\texttt R}{\texttt L}{\texttt E}} 
\newcommand{\dd}{\mathinner{\ldotp\ldotp}}

\newtheorem{observation}{Observation}
\newtheorem{fact}{Fact}
\newtheorem{myclaim}{Claim}

\begin{document}

\algdef{SE}[SWITCH]{Switch}{EndSwitch}[1]{\algorithmicswitch\ #1\ \algorithmicdo}{\algorithmicend\ \algorithmicswitch}%
\algdef{SE}[CASE]{Case}{EndCase}[1]{\algorithmiccase\ #1}{\algorithmicend\ \algorithmiccase}%
\title{Lyndon Border Arrays and Lyndon Suffix Arrays}
\title{Linear Algorithms for Computing the \emph{Lyndon Border Array} and
the \emph{Lyndon Suffix Array}}
\titlerunning{\textit{Lyndon Border Arrays and Lyndon Suffix Arrays}}

\author[Alatabbi\textit{et~al}]{
\textbf{
Ali Alatabbi\,$^{1}$,
Jacqueline W.\ Daykin\,$^{1,2}$,
and
M. Sohel Rahman\,$^{3}$
}
}

\institute{
\small
$^{1}$Department of Informatics, King's College London, London, UK.\\
$^{2}$ Royal Holloway, University of London, UK,\\
$^{3}$A$\s{\ell}$EDA Group, Department of CSE, BUET, Dhaka-1000, Bangladesh
}

\authorrunning{\textit{Alatabbi, Daykin and  Rahman.}}
\maketitle

\begin{abstract}
\ifjVer 
We consider the problem of finding repetitive structures and inherent
patterns in a given string $\s{s}$ of length $n$ over a finite totally ordered alphabet.
A border $\s{u}$ of a string $\s{s}$ is both a prefix and a suffix of
$\s{s}$ such that $\s{u} \not= \s{s}$.
The computation of the border array of a string $\s{s}$, namely the borders of each prefix
of $\s{s}$, is strongly related to the string matching problem: given a string
$\s{w}$, find all of its occurrences in $\s{s}$.

A {\itshape Lyndon word} is a primitive word (i.e., it is not a power of another word) which is minimal for the lexicographical order of its conjugacy class (i.e., the set of words obtained by cyclic rotations of the letters).

In this paper we combine these concepts to introduce the \emph{Lyndon
Border Array} $\mathcal L \beta$ of $\s{s}$, whose $i$-th entry $\mathcal L
\beta(\s{s})[i]$ is the length of the longest border of $\s{s}[1 \dd i]$
which is also a Lyndon word. We propose linear-time and linear-space
algorithms \footnote{
The algorithms presented in this paper can be applied to computing the \emph{co-Lyndon Border
Array} by observing that a word $\s{u} = u_1 u_2 \cdots u_n$ is a co-Lyndon word if the reversed word $\s{u} = u_n u_{n-1} \cdots u_1$ is a Lyndon word.
}
for computing $\mathcal L \beta (\s{s})$.
Further, we introduce the \emph{Lyndon Suffix Array}, and by modifying the
efficient suffix array technique of Ko and Aluru \cite{KA03} outline a
linear time and space algorithm for its construction.
\fi

\end{abstract}
\section{Introduction}\label{lba:sec:intro}
Understanding complex patterns and repetitive structures in strings is essential for efficiently solving many problems in stringology \cite{CR02}.
For instance, Lyndon words are increasingly a fundamental and applicable form in the study of combinatorics on words \cite{L83}, \cite{L05}, \cite{S03} - these patterned words have deep links with algebra and are rich in structural properties. Another important concept is a border $\s{u}$ of a string $\s{s}$ defined to be both a prefix and a suffix of $\s{s}$ such that $\s{u} \not= \s{s}$.
The computation of the border array of a string $\s{s}$, that is of the borders of each prefix
of $\s{s}$, is strongly related to the string matching problem: given a string
$\s{w}$, find all of its occurrences in a string $\s{s}$.
It constitutes the ``failure function'' of the Morris-Pratt (1970) string matching algorithm \cite{MORRIS-PRATT1970}.

Lyndon words were introduced under the name of \emph{standard lexicographic
sequences} \cite{Lyn54,Lyn55} in order to construct a basis of a free
Abelian group. Two strings are \emph{conjugate} if they differ only by a
cyclic permutation of their characters; a \emph{Lyndon word} is defined as
a (generally) finite word which is strictly minimal for the lexicographic
order of its conjugacy class.  For a non-letter Lyndon word \s{w}, the pair
(\s{u}, \s{v}) of Lyndon words such that \s{w} = \s{uv} with \s{v} of
maximal length is called the \textit{standard factorization} of $\s{w}$.

The set of Lyndon words permits the unique maximal factorization of any
given string \cite{CFL58,L83}. In 1983, Duval \cite{Du83} developed an
algorithm for standard factorization that runs in linear time and space --
the algorithm cleverly iterates over a string trying to find the longest
Lyndon word; when it finds one, it adds it to the result list and proceeds
to search in the remaining part of the string.


Lyndon words proved to be useful for constructing bases in free Lie
algebras \cite{R93}, constructing de Bruijn sequences \cite{Necklaces78},
computing the lexicographically smallest or largest substring in a string
\cite{ApostolicoC95}, succinct suffix-prefix matching of highly periodic
strings \cite{NeuburgerS13}.
Wider ranging applications include the Burrows-Wheeler transform and data
compression \cite{GS-09}, musicology \cite{C04}, bioinformatics
\cite{DR-04}, and in relation to cryptanalysis \cite{P05}.
Indeed the uses, and hence importance, of Lyndon words are increasing, and
so we are motivated to investigate specialized Lyndon data structures.\\

The key contributions of this paper are as follows.
\begin{itemize}
\item By combining the important concepts of Lyndon words and borders of strings, we introduce here the \emph{Lyndon Border Array} $\mathcal L \beta$ of $\s{s}$, whose $i$-th entry $\mathcal L \beta(\s{s})[i]$ is the length of the longest border of $\s{s}[1 \dd i]$ which is also a Lyndon word.
We present an efficient linear time and space algorithm for computing the \emph{Lyndon Border Array} $\mathcal L \beta$ for a given string (Section~\ref{lba:sec:algorithm}).

\item In order to achieve the desired level of efficiency in the Lyndon Border Array construction we also present some interesting results related to Lyndon combinatorics, which we believe is of independent interest as well (Section~\ref{lba:sec:combinatorics}).

\item A complementary data structure, the \emph{Lyndon Suffix Array}, which is an adaptation of the classic suffix array, is also defined; by modifying the linear-time construction of Ko and Aluru \cite{KA03} we similarly achieve a linear construction for our Lyndon variant (Section~\ref{lba:sec:suffix}). We also present a simpler algorithm to construct a Lyndon Suffix Array from a given Suffix Array (Section~\ref{lba:subsec:simpleLSA}). The latter algorithm also runs in linear time and space.
\end{itemize}


\section{Basic Definitions and Notation}\label{lba:sec:preliminaries}
Consider a finite totally ordered alphabet $\Sigma$ which consists of a set of characters (equivalently letters or symbols).
The cardinality of the alphabet is denoted by $|\Sigma|$. 

A string (word) is a sequence of zero or more characters over an alphabet
$\Sigma$.
A string $\s{s}$ of length $|\s{s}|=n$ is represented by $\s{s}[1 \dd n]$, where $\s{s}[i] \in \Sigma $ for $ 1\leq i \leq n $.
The set of all non-empty strings over the alphabet $\Sigma$ is denoted by $\Sigma^+$, and the set of strings of length $n$ by $\Sigma^n$.
The empty string is the empty sequence of characters (with zero length) denoted by $\epsilon$, with
$\Sigma^* = \Sigma^+ \cup \epsilon$; we write all strings in mathbold: \s{u}, \s{v}, and so on.

The $i$-th symbol of a string $\s{s}$ is denoted by $\s{s}[i]$, or simply $s_i$. We denote by $\s{s}[i \dd j]$, or $s_i \cdots s_j$, the substring of $\s{s}$ that starts at position $i$ and ends at position $j$.

A string $\s{w}$ is a substring, or factor, of $\s{s}$ if $\s{s} = \s{u}\s{w}\s{v}$, where $\s{u},\s{v} \in \Sigma^{\ast}$; specifically, a string $\s{w} = w_1 \cdots w_m$ is a substring of $\s{s} = s_1 \cdots s_n$ if $w_1 \cdots w_m = s_i \cdots s_{i+m-1}$ for some $i$.
Words $\s{w}[1 \dd i]$ are called prefixes of $\s{w}$, and words
$\s{w}[i \dd n]$ are called suffixes of $\s{w}$. The prefix $\s{u}$
(respectively suffix $\s{v}$) is a proper prefix (respectively  suffix) of
a word $\s{w}=\s{u}\s{v}$ if $\s{w}\not=\s{u},\s{v}$.

For a substring $\s{w}$ of $\s{s}$, the string $\s{u}\s{w}\s{v}$ for $\s{u}, \s{v} \in \Sigma^*$ is an extension of $\s{w}$ in $\s{s}$ if $\s{u}\s{w}\s{v}$ is a substring of $\s{s}$; $\s{w}\s{v}$ for $\s{v} \in \Sigma^*$ is the right extension of $\s{w}$ in $\s{s}$ if $\s{w}\s{v}$ is a substring of $\s{s}$; $\s{u}\s{w}$ for $\s{u} \in \Sigma^*$ is a left extension of $\s{w}$ in $\s{s}$ if $\s{u}\s{w}$ is a substring of $\s{s}$. Words that are both prefixes and suffixes of $\s{w}$ are called borders of $\s{w}$. By $border(\s{w})$ we denote the length of the longest border of $\s{w}$ that is shorter than $\s{w}$.


A word $\s{w}$ is periodic if it can be expressed as $\s{w}=\s{p}^k\s{p}^\prime$ where $\s{p}^\prime$ is a proper prefix of $p$, and $k \geq 2$.
Moreover, a string is said to be primitive if it cannot be be written as $\s{u}^k$ with $\s{u} \in \Sigma^+$ and $k \geq 2$, i.e., it is not a power of another string. When $\s{p}$ is primitive,
we call it ``the period" of $\s{u}$. It is a known fact \cite{DBLP:books/daglib/0020103} that, for any string $\s{w}$, $per(\s{w}) + border(\s{w}) = |\s{w}|$, where the period {\emph per} of a nonempty
string is the smallest of its periods.

\begin{definition}{(Border array)}
For a string $\s{s} \in \Sigma^n$, the border array $\beta(\s{s}) [1 \dd n]$ is defined by
$\beta(\s{s})[i] = |border(\s{s}[1 \dd i])|$ for $1 \leq i \leq n$.
\end{definition}

\begin{proposition}\label{prop:border-array}
\cite{MORRIS-PRATT1970} The border of a string $\s{s}$ (or the table
$\beta(\s{s})$ itself) can be computed in time $O(|\s{s}|)$.
\end{proposition}

A string $\s{y}=\s{y}[1 \dd n]$ is a \emph{conjugate} (or cyclic rotation) of $\s{x}=\s{x}[1 \dd n]$ if $\s{y}[1 \dd n] = \s{x}[i \dd n]\s{x}[1 \dd i-1]$ for some $1 \leq i \leq n$ (for $i =1, ~\s{y} =\s{x}$).
A {\itshape Lyndon word} is a primitive word
which is minimal for the lexicographical order of its conjugacy class
(i.e., the set of all words obtained by cyclic rotations of letters).
Furthermore, a non-empty word is a Lyndon word if and only if it is
strictly smaller in lexicographical order (lexorder) than any of its
non-empty proper suffixes \cite{Du83,L83}.

Throughout this paper,
$\mathcal{L}$  will denote the set of Lyndon words over the totally ordered alphabet $\Sigma$, $\mathcal L_n$ will denote the set of Lyndon words of length $n$; hence $\mathcal L = \{\mathcal L_1 \cup \mathcal L_2 \cup \mathcal L_3,\dd \}$. We next list several
well-known properties of Lyndon words and border arrays which we later apply to develop the new algorithms.

\begin{proposition}
\cite{Du83} A word $\s{w} \in \Sigma^+$ is a Lyndon word if and only if either $\s{w} \in \Sigma$ or $\s{w} = \s{u}\s{v}$ with $\s{u}$, $\s{v} \in \mathcal{L}$, $\s{u} < \s{v}$.
\end{proposition}

\begin{theorem}\label{lLyndon-thrm}
\cite{CFL58} Any word $\s{w}$ can be written uniquely as a non-increasing product $\s{w} = \s{u}_1 \s{u}_2 \cdots \s{u}_k$ of Lyndon words.
\end{theorem}

Theorem \ref{lLyndon-thrm} shows that there is a unique decomposition of any word into non-increasing Lyndon words $(\s{u}_1 \ge \s{u}_2 \ge \cdots \ge \s{u}_k)$.

\begin{observation}\label{obs:Lyndon-Word}
Let $\s{\ell}$  be a Lyndon word ($\s{\ell} \in \mathcal{L}$) where $\s{\ell} = \ell_1 \ell_2 \dd  \ell_n$ (to avoid trivialities we assume $n>1$), then
\begin{enumerate}[(1)]
 \item $border(\s{\ell}) =0$, \label{obs:Lyndon-Word:primitive}
 \item $\ell_1 < \ell_n$, \label{obs:Lyndon-Word:unit-factor}
 \item $\ell_1 \leq \ell_i|\ell_i \in \{\ell_2, \dd , \ell_{n-1}\}$, \label{obs:Lyndon-Word:prefix}
 \item $\s{\ell} < \ell_i \cdots \ell_n$, for $ 1 < i \le n$.\label{obs:Lyndon-Word:suffix}
\end{enumerate}
\end{observation}
\begin{observation}\label{obs:Border-Array}
Given a string $\s{s}$, then
\begin{enumerate}[(1)]
 \item $\beta(\s{s})[1] =0$,  \label{obs:first_element}
 \item if $\s{b}$ is a border of $\s{s}$, and $\s{b}^\prime$ is a border of $\s{b}$, then $\s{b}^\prime$ is a border of $\s{s}$,
 \item $0 \leq \beta (\s{s}) [i+1] \leq \beta(\s{s})[i] +1$, for $ 1 \le i < n$.
\end{enumerate}
\end{observation}

We now introduce the \emph{Lyndon Border Array} and associated computation, illustrated in Example \ref{example01} below.

\begin{definition}\label{definition-LBA}
(Lyndon Border Array)
For a string $\s{s} \in \Sigma^n$, the \emph{Lyndon border array} $\mathcal L \beta(\s{s})[i]$ is the length of the longest border
 of $\s{s}[1 \dd i]$ which is also a Lyndon word.
\end{definition}


\begin{definition}\label{definition-LSA}
(Lyndon suffix array)
For a string $\s{s} \in \Sigma^n$, the \emph{Lyndon Suffix Array} of \s{s} is the lexicographically sorted list of all those suffixes of \s{s} that form Lyndon words.
\end{definition}

Given a string $\s{s}$ of length $n$, associated computational problems are:
compute the Lyndon border and Lyndon suffix arrays; we address these problems in this paper.

\begin{example}\label{example01}
Consider the string $\s{s}=a b a a b a a a b b a a b a a b$.
The following table illustrate the border array $\beta$ of \s{s},
the Lyndon border array $\mathcal L \beta$ of $\s{s}$,
the suffix array  $\mathcal A$ of $\s{s}$ and
the Lyndon suffix array  $\mathcal L S$ of $\s{s}$.

\addtolength{\textfloatsep}{-7mm}
\setlength{\tabcolsep}{5pt}
\begin{table}[h]
\center 
\begin{tabular}{ r | c c c c c c c c c c c c c c c c}
 
$\scriptstyle i$ & $\scriptstyle  0$& $\scriptstyle 1$& $\scriptstyle 2$&
$\scriptstyle 3$& $\scriptstyle 4$& $\scriptstyle 5$& $\scriptstyle 6$&
$\scriptstyle 7$& $\scriptstyle 8$& $\scriptstyle 9$& $\scriptstyle 10$&
$\scriptstyle 11$ & $\scriptstyle 12$ & $\scriptstyle 13$ & $\scriptstyle
14$ & $\scriptstyle  15$\\
\hline
$\s{s}[i]$ & a &b &a &a &b &a &a &a &b &b &a &a &b &a &a &b \\
$\beta[i]$ & 0 &0 &1 &1 &2 &3 &4 &1 &2 &0 &1 &1 &2 &3 &4 &5 \\
$\mathcal L \beta[i]$ & 0 &0 &1 &1 &2 &1 &1 &1 &2 &0 &1 &1 &2 &1 &1 &2 \\
$\mathcal A [i]$ & 5&13&2&10&6&14&3&11&0&7&15&4&12&1&9&8 \\
$\mathcal LS [i]$ &5& 13& 14& 15 & -1 &-1 &-1 &-1 &-1 &-1 &-1 &-1 &-1 &-1
&-1 &-1  \\
 
\end{tabular}
\hspace{\textwidth}
\label{tab:ex1}
\end{table}
\end{example}



\section{Lyndon Combinatorics}\label{lba:sec:combinatorics}

This section introduces some new interesting combinatorial results on Lyndon words. In relation to the computation of the Lyndon Border Array, we here show how to find the shortest prefix of a string that is both border-free and not a Lyndon word.
So assume that for a given string $\s{s}$ of length $n$, we have $\s{s}[1] = \gamma$.  If $\s{f_1}, \dd , \s{f_q}$ are factors of $\s{s}$, we use $start(\s{f_i})~(end(\s{f_i}))$ to denote the index of $\s{f_i}[1]~(\s{f_i}[|\s{f_i}|])$ in $\s{s}$, and say that $j$ is an index of $\s{f_i}$ if $start(\s{f_i})\leq j \leq end(\s{f_i})$. An outline of the steps of the algorithm is as follows:\\

\noindent Algorithm \emph{Shortest non-Lyndon Border-free Prefix (SNLB$f$P)}.
\begin{enumerate}
\item\label{step:LynFact} Compute the Lyndon factorization of $\s{s}$.
\item\label{step:BinSearch} Apply binary search to find the first Lyndon factor $\s{f_{\mu}}$ in the factorization starting with the largest letter $\mu$ which is strictly less than $\gamma$ (if it exists).
\item\label{step:border} Consider the maximal prefix $\s{p}$ of $\s{s}$ in which every factor $\s{f_1}, \dd , \s{f_q}$ starts with $\gamma$; compute the border array $\beta(\s{p})$ of $\s{p}$.
\item\label{step:computeI} Compute $i$, the smallest index of $\s{p}$, such that $i>end(\s{f_1})$ ($i$ is not an index of $\s{f_1}$) and $\beta(\s{p})[i] = 0$;
\begin{enumerate}
\item if $i$ does not exist then $i = end(\s{f_q}) +1$ (if it exists)
\item if $q = 1$ then $i = end(\s{f_1}) +1$ (if it exists).
\end{enumerate}
\item Return $\s{s}[1 \dd i]$. 
\end{enumerate}

\begin{myclaim}\label{claim:F1}
Suppose $\s{s}[1 \dd i]$ is the shortest prefix of $\s{s}$ that is both
border-free and non-Lyndon. Then $i > end(\s{f_1})$, i.e., 
Algorithm SNLB$f$P is correct in skipping the first Lyndon factor.
\end{myclaim}
\begin{proof}[Claim~\ref{claim:F1}]
The proof is by induction.
Assume that the first Lyndon factor $\s{f_1}$ is of length $m$ (with $m \ge 2$ otherwise the Claim holds trivially). By Observation \ref{obs:Lyndon-Word}(\ref{obs:Lyndon-Word:primitive}) and Observation \ref{obs:Border-Array}(\ref{obs:first_element}) we have $\beta(\s{f_1})[1] = \beta(\s{f_1})[m] = 0$. The smallest $j$ after 1 where $\beta(\s{f_1})[j] = 0$ must index a Lyndon word \s{x} of the form $\s{x} = \gamma^{j-1}\nu$, where $\nu > \gamma$.\ Hence, after the first letter (which is a Lyndon word), the next border-free prefix is also a Lyndon word.


Assume now that all border-free prefixes up to index $t$ of $\s{f_1}$ are
Lyndon words (and hence nested), and suppose that the next border-free
position is $t'$. Then we need to show that $\s{y} = \s{f_1}[1 \dd t']$
is a Lyndon word; we proceed to show that \s{y} is less than each of its
proper suffixes. Let $\s{w_k} = \s{f_1}[t'-k+1 \dd t']$ for $1 \le k
<t'$. Since \s{f_1} is a Lyndon word and minimal in its conjugacy class,
then $\s{f_1}[1 \dd k] \le \s{w_k}$; further, since $\beta(\s{f_1})[t']
= 0$  we cannot have equality and so $\s{f_1}[1 \dd k] < \s{w_k}$ which
implies that $\s{y} < \s{w_k}$ as required.


In other words, we have shown that any border-free prefix of $\s{f_1}$ is a
Lyndon word and the result follows.

\qed
\end{proof}

\begin{corollary}\label{cor:Lyndonprefix}
Any border-free prefix of a Lyndon word is a Lyndon word.
\end{corollary}

\begin{lemma}\label{algo:snlpCor}
\emph{Algorithm SNLB$f$P} is correct.
\end{lemma}
\begin{proof}[Lemma~\ref{algo:snlpCor}]
In Step~\ref{step:BinSearch} we identify the factor $\s{f_{\mu}}$ which starts with the letter $\mu < \gamma$. From Lyndon principles, Observation \ref{obs:Lyndon-Word}(\ref{obs:Lyndon-Word:unit-factor}),(\ref{obs:Lyndon-Word:prefix}),
it follows that no factor to the left of $\s{f_{\mu}}$ contains the letter $\mu$. Hence the prefix $\s{s}[1 \dd start(\s{f_{\mu}})]$ is both border-free and non-Lyndon. However, it may not be the shortest one and so the algorithm continues to check through the prefix $\s{p}$.

Now, consider the index $i$ of $\s{p}$ computed in
Step~\ref{step:computeI}. Suppose that $i$ is an index of the factor
$\s{f_t}$; by Claim \ref{claim:F1} we have $t>1$. Let $k$ be the length of
the prefix $\s{p_t}$ of $\s{f_t}$ that ends at $i$, and $\s{p_1}$ be the
prefix of $\s{f_1}$ of length $k$. By the Lyndon factorization we have
$\s{p_1} \ge \s{p_t}$ (in lexicographic order). If $\s{p_1} = \s{p_t}$ then
this contradicts $\beta(\s{p})[i] = 0$. Hence $\s{p_1} > \s{p_t}$ and so
the prefix $\s{s}[1 \dd i]$ is both border-free and not a Lyndon word.
Hence Algorithm SNLB$f$P correctly returns $\s{s}[1 \dd i]$.

\qed
\end{proof}

\begin{lemma}\label{algo:snlpRT}
\emph{Algorithm SNLB$f$P} runs in $O(n)$ time.
\end{lemma}
\begin{proof}[Lemma~\ref{algo:snlpRT}]
Step~\ref{step:LynFact} can be computed in $O(n)$ time \cite{Du83,L05}. Step~\ref{step:BinSearch} applies an $O(\log n)$ binary search.
 In Step~\ref{step:border} we compute the border array of the prefix $\s{p}$ of $\s{s}$. Clearly Steps \ref{step:border} and \ref{step:computeI} can be completed in $O(n)$ time. Hence, the result follows.

\qed
\end{proof}

\begin{lemma}\label{lemm:Lyndon-invalid-point-primitive}\textbf{(Lyndon invalid
point for border-free word)} Given a string $\s{\ell} \in \Sigma^n$, $n>1$,
such that $border(\s{\ell}) =0$,
if $\s{\ell}$ is not a Lyndon word ($\s{\ell} \not \in \mathcal L$), then all
the right extensions $\s{\ell}\s{\ell}^\prime$ of $\s{\ell}$, such that
$\s{\ell}^\prime \in \Sigma^+$, are not Lyndon words either.

--- We refer to this condition as the $\mathcal{L}$-fail condition and to the
point (index) of where it occurs as the Lyndon invalid point.
\end{lemma} 

\begin{proof}[Lemma~\ref{lemm:Lyndon-invalid-point-primitive}]
Since $\s{\ell}$ is border-free ($border(\s{\ell}) =0$) but not a Lyndon word,
then let $\s{r}$ be the rotation of $\s{\ell}$ which is the Lyndon word of the
conjugacy class. Write $\s{\ell} = \s{p}\s{q}$ such that $\s{r} = \s{q}\s{p}$
and $\s{q}\s{p} < \s{p}\s{q}$.

\begin{description}
 \item [Case (1) -] If $|\s{p}| = |\s{q}|$,
 then since $\s{r}$ is border-free and a Lyndon word, $\s{q} \neq \s{p}$.
 Further, since $\s{q}\s{p} < \s{p}\s{q}$, we have $\s{q} < \s{p}$.
 It follows that $\s{q}\s{\s{\ell}'}\s{p} < \s{p}\s{q}\s{\s{\ell}'}$ and so $\s{\ell}\s{\s{\ell}'}$ cannot be a Lyndon word.

 \item[Case (2) -] If $|\s{q}| > |\s{p}|$,
 then $\s{r}[1 .. |\s{p}|] < \s{p}$ (i.e., $\s{q}[1 \dd |\s{p}|] < \s{p}$),
 we have $\s{q} < \s{p}$.
 It follows that $\s{q}\s{\s{\ell}'}\s{p} < \s{p}\s{q}\s{\s{\ell}'}$ and so
 $\s{\ell}\s{\s{\ell}'}$ cannot be a Lyndon word.

 \item[Case (3) -] If $|\s{q}| < |\s{p}|$. We have $\s{r} = \s{q}\s{p} \in
 \mathcal L$, where $\s{q} = \s{q_1} \s{q_2} \cdots \s{q_j}$, $\s{p} = \s{p_1}
 \s{p_2} \cdots \s{p_k}$. We are required to show that $\s{r^\prime} =
 \s{p}\s{q}\s{\ell}^\prime \notin \mathcal L$; so suppose that $\s{r}^\prime
 \in \mathcal L$.
From $\s{r}$ we have that $\s{q}_1 \le \s{p}_1$, while from $\s{r}^\prime$ we
have $\s{p}_1 \le \s{q}_1$, which together implies $\s{q}_1 = \s{p}_1$. From
the rotation of $\s{r}$ starting $\s{p}_1 \s{p}_2$ we have $\s{q}_2 \le
\s{p}_2$, while from the rotation of $\s{r}^\prime$ starting $\s{q}_1 \s{q}_2$
we find that $\s{p}_2 \le \s{q}_2$ giving $\s{q}_2 = \s{p}_2$. We continue this
argument for $|\s{q}|$ elements which shows that the given word $\s{p}\s{q}$
has a border.
\end{description} 

In the first two cases the order is decided within the first $|\s{p}|$
elements.

\qed
\end{proof}

\begin{lemma}\label{lemm:Lyndon-invalid-point-bordered}\textbf{(Lyndon invalid point for bordered word)}
Given a string $\s{\ell} \in \Sigma^n$,
$n>1$, such that $border(\s{\ell}) >0$,
let $\{\s{\ell^{\prime1}},\s{\ell^{\prime2}},\s{\ell^{\prime3}} \cdots \}$ be right extensions of $\s{\ell}$ by $\{1,2,3,\cdots\}$ characters in $\Sigma$ respectively. Then $\s{\ell}$ and all its right extensions are not Lyndon words if $\s{\ell^{\prime1}}[|\s{\ell^{\prime1}}|] < \s{\ell^{\prime1}}[border(\s{\ell})+1]$.

--- Similarly to Lemma \ref{lemm:Lyndon-invalid-point-primitive}, we refer to this condition as the $\mathcal{L}$-fail condition and to the point (index) of where it occurs as the Lyndon invalid point.
\end{lemma}

\begin{proof}[Lemma~\ref{lemm:Lyndon-invalid-point-bordered}]
The lemma follows from the following two cases:

\begin{enumerate} 
 \item $\s{\ell} \notin \mathcal L$, this is immediate from the hypothesis that
 $border(\s{\ell})>0$ and Observation \ref{obs:Lyndon-Word}(1).
\item Consider the right extension $\s{\ell^{\prime m}}$ of $\s{\ell}$,
where $m \ge 1$. Then the suffix $\s{\ell{^{\prime m}}}[n -
border(\s{\ell})+1 \dd n+m]$ of $\s{\ell^{\prime m}}$ is lexicographically
less than $\s{\ell}^{\prime m}[1 \dd border(\s{\ell})+1]$ and
consequently less than $\s{\ell{^{\prime m}}}$, contradicting the property that a 
Lyndon word is strictly smaller than any of its proper suffixes. Hence no
right extension of $\s{\ell}$ is a Lyndon word.
\end{enumerate}

In case (2) the order is decided within the first $border(\s{\ell}) + 1$ elements.

\qed 
\end{proof}  

\begin{fact}
For a given string $\s{s} \in \Sigma^n$,
suppose we have computed $\beta(\s{s})$. Then, for $1 \leq i \leq n$, the following holds true:
$$\mathcal L\beta(\s{s})[i] \in \{\beta(\s{s})[i],~\beta(\s{s})[\beta(\s{s})[i]],~\beta(\s{s})[\beta(\s{s})[\beta(\s{s})[i]]],~ \dd, ~ 0 \}.
$$
\end{fact}

A binary Lyndon word can also be expressed in terms of Lyndon properties of the
integer parameters (exponents) given by its \emph{Run Length Encoding}. For a
binary string $\s{\ell}$, let $\RLE(\s{\s{\ell}_a})$ denote the
encoding (as a string) of the subsequence of $\s{\ell}$ consisting of all letters $a$
but no letter $b$ ($p^\prime, p_1; \dd; p_m$), similarly
$\RLE(\s{\s{\ell}_b})$ denotes the encoding of the subsequence of
$\s{\ell}$ consisting of all letters $b$ ($q\prime, q_1; \dd; q_m$).

\begin{proposition}\label{prop:rle}
For a given string $\s{s} \in \Sigma^+$, the \emph{Run Length Encoding}
$RLE(\s{s})$ and subsequently, the list $\mathcal R(\s{s})$ can be computed
in time and space linear in the size of the given string $\s{s}$.
\end{proposition}

Now we have the following Proposition to check if a given binary word is a
Lyndon word or not (using the \emph{Run Length Encoding} of the binary
word).

\begin{proposition}\label{prop:Lyndon-border-binary}
Let $\mathcal{L}$ be the set of Lyndon words over an alphabet $\Sigma$,
where $\Sigma = \{a,b\}$ and $a < b$.
For a non-letter word $\s{\ell} = \s{\ell}_1\s{\ell}_2 \cdots \s{\ell}_n$
and its corresponding exponents list $\mathcal R(\mathcal L)$, we have the
following:


\begin{description}
 \item [Case $ m = 0$,] we have $p_j, q_j = 0$ for $j \in \{1 \dd m\}$ then
  $\s{\ell} \in \mathcal{L}$ if and only if $p^\prime, q^\prime > 0$.

\item [Case $ m > 0$,] we have $p^\prime, q^\prime, p_j, q_j > 0$  then
   \begin{enumerate}
    \item $\s{\ell} \in \mathcal{L}$ if and only if $p^\prime > p_j$ for $j
    \in \{1 \dd m\}$.

    \item if $p^\prime = p_m$ and $p^\prime > p_j$ for $j \in \{1 \dd
    m-1\}$ then $\s{\ell} \in \mathcal{L}$ if and only if $q^\prime < q_m
    \wedge \lambda < 0$.

    \item if there exists $p_j$ such that $p^\prime =p_j$ and $p^\prime >
    p_m$ then $\s{\ell} \in \mathcal{L}$ if and only if $q^\prime \le q_j
    \wedge \lambda < 0$ for $j \in \{1 \dd m-1\}$. 

   \end{enumerate}
\end{description}

where $\lambda$ is the index of Lyndon invalid point i.e.
$\mathcal{L}$-fail condition (As defined in
Lemma.~\ref{lemm:Lyndon-invalid-point-primitive} and Lemma. \ref{lemm:Lyndon-invalid-point-bordered}). 

\noindent
\end{proposition}

\begin{lemma}\label{lemm:Lyndon-encoding-binary}
Let $\s{\ell}$ be a binary word with $\s{\ell}[1]=a$, $\s{\ell}[n]=b$ and
associated encodings $\RLE(\s{\s{\ell}_a})$ and
$\RLE(\s{\s{\ell}_b})$. Then $\s{\ell}$ is a Lyndon word if and only
if, either \\
(i) $\RLE(\s{\s{\ell}_a})$ is a Lyndon word on the alphabet $\{1>2>3>
\cdots \}$, or \\
(ii) $\RLE(\s{\s{\ell}_a})$ is a repetition of a Lyndon word as in (i)
and $\RLE(\s{\s{\ell}_b})$ is a Lyndon word on the alphabet $\{1<2<3<
\cdots \}$.
\end{lemma}

\begin{proof}[Lemma~\ref{lemm:Lyndon-encoding-binary}]
First we consider necessity.
So suppose that (i) holds. Consider any rotation $\s{\s{\ell}_a^r}$ of
$\s{\s{\ell}_a}$ (including those with split runs of $a$'s). Then
$\RLE(\s{\s{\ell}_a}) < \RLE(\s{\s{\ell}_a^r})$ in lexorder
over $\{1>2>3> \cdots \}$. Now suppose that (ii) holds. Then for a rotation
$\s{\s{\ell}_a^r}$ of $\s{\s{\ell}_a}$, either $\RLE(\s{\s{\ell}_a}) <
\RLE(\s{\s{\ell}_a^r})$ in lexorder over $\{1>2>3> \cdots \}$, or
$\RLE(\s{\s{\ell}_a}) = \RLE(\s{\s{\ell}_a^r})$ and
$\RLE(\s{\s{\ell}_b}) < \RLE(\s{\s{\ell}_b^r})$ in lexorder
over $\{1<2<3< \cdots \}$.

For sufficiency, the conditions guarantee that $\s{\ell} \in \mathcal{L}$.

\qed
\end{proof}

Lemma~\ref{lemm:Lyndon-encoding-binary} shows that a binary string can be
decomposed into its Lyndon factors with a double application of Duval's
linear factorization algorithm \cite{Du83}, first on the exponents of $a$
and then on those for $b$ -- hence linear in $|\mathcal{R}(\s{\s{\ell}})|$
(once $\mathcal{R}(\s{\s{\ell}})$ is constructed in time linear in
$|\ell|$). It follows that a word can be tested to be a Lyndon word in
linear time: check whether there is more than one factor in the
factorization.

\begin{lemma}\label{lemm:Lyndon-prefix-binary}
Suppose we are given a binary string $\s{s}$ of length $n$ and the
corresponding list of exponents $\mathcal R(\s{s})$. Then we can compute
the array $\Psi(\s{s})$ in $O(n)$ time.
\end{lemma}

\begin{proof}[Lemma~\ref{lemm:Lyndon-prefix-binary}]
We first focus on computing the $i$-th element $\Psi[i](\s{s})$ of the
array $\Psi(\s{s})$. Clearly, by Proposition
\ref{prop:Lyndon-border-binary}, it takes constant number of comparisons to
determine whether the prefix $\s{s}[1 \dd i]$ is a Lyndon word or not.
Since we have to repeat the steps in Proposition
\ref{prop:Lyndon-border-binary} for $n$ prefixes, therefore $\Psi(\s{s})$
can be done in linear time once $\mathcal R(\s{s})$ is computed.

\qed
\end{proof}

\section{Lyndon Border Array Computation}\label{lba:sec:algorithm}
In this section we develop an efficient algorithm for computing the Lyndon
Border Array. We first recall an interesting relation that exists for
borders which we refer to as the \emph{Chain of Borders} henceforth. Since
every border of any border of $\s{s}$ is also a border of $\s{s}$ 
(Observation 2(2)), it turns out that, the border array $\beta(\s{s})$
compactly describes all the borders of every prefix of $\s{s}$. For every
prefix $\s{s}[1 \dd i]$ of $\s{s}$, the following sequence
\begin{equation}\label{seq:cover_array} \beta(\s{s})^{1}[i],
\beta(\s{s})^{2}[i], \dd , \beta(\s{s})^{m}[i]
\end{equation}
is well defined and monotonically decreasing to $\beta(\s{s})^{m}[i] = 0$
for some $m \ge 1$ and this sequence identifies every border of $\s{s}[1
\dd i]$. Here, $\beta(\s{s})^{k}[i]$ is the length of the $k$-th longest
border of $\s{s}[1\dd i],$ for $1\leq k\leq m$. Sequence
(\ref{seq:cover_array}) identifies the above-mentioned chain of borders.
We will also be using the usual notion of the length of the chain of
borders. Clearly, the length of the chain in Sequence
(\ref{seq:cover_array}) is $m$. Also we use the following notion and
notations. In Sequence (\ref{seq:cover_array}), we call
$\beta(\s{s})^{m}[i] = 0$ the \emph{last} value and
$\beta(\s{s})^{m-1}[i]$ the \emph{penultimate} value in the chain. We now
present the following interesting facts that will be useful in our algorithm.
\begin{fact}\label{fact:noChain}
The length of the chain of borders for a Lyndon Border is at most 2.
\end{fact}
\begin{proof}
The result follows, because, if a border is a Lyndon Word, then it cannot
itself have a border (Observation 1(1)).

\qed
\end{proof}
\begin{fact}\label{fact:xAnd0}
Suppose $\mathcal L \beta(\s{s})[i] = x$. Then $\mathcal L \beta(\s{s})[x] = 0$.
\end{fact}
\begin{proof}
Immediate from Fact \ref{fact:noChain}.

\qed
\end{proof}
Now we are ready to propose a straightforward naive algorithm to compute
the Lyndon Border Array $\mathcal L\beta(\s{s})$ for $\s{s}[1\dd n]$ as
follows.

\noindent Algorithm \emph{Naive Lyndon Border Array Construction}
\begin{algorithmic}[1]
\State Compute $\beta(\s{s})[1\dd n]$
\For{$ i = n  \to 1$}
    \State Find the penultimate value $k$ of the chain of borders for $\beta(\s{s})[i]$
    \If{$\s{s}[1\dd k]$ is a Lyndon word \label{Step:LyndonCheck}}
         \State Set $\mathcal L\beta(\s{s})[i] = k$
    \Else
        \State Set $\mathcal L\beta(\s{s})[i] = 0$
    \EndIf
\EndFor
\end{algorithmic}

The correctness of the algorithm follows directly from
Facts~\ref{fact:noChain} and \ref{fact:xAnd0}. Now we discuss an efficient
implementation of the algorithm. When we traverse through the chain of
borders, we reach the penultimate value $k$ and then the last value $0$.
Clearly, all we need is to check for each such chain, whether
$\s{s}[1\dd k]$ is a Lyndon word. So, we are always interested in
finding whether $\s{s}[1\dd k]$ is a Lyndon word where $\beta(\s{s})[k]
= 0$. At this point the computation of SNLB$f$P
(Section~\ref{lba:sec:combinatorics}) will be applied. To give an example,
by Corollary~\ref{cor:Lyndonprefix}, we know that any border-free prefix to
the left of SNLB$f$P is a Lyndon word and any border-free prefix to the
right of SNLB$f$P is non-Lyndon. This gives us an efficient weapon to check
the \emph{If} statement of Line~\ref{Step:LyndonCheck} of the above
algorithm. Finally, as can be seen below, we can make use of a stack data
structure along with some auxiliary arrays to efficiently implement the
above algorithm. In particular, we simply keep an array $Done[1\dd n]$
initially all false, using a stack and the SNLB$f$P index (say $r$). The
algorithm is presented below:

\noindent Algorithm \emph{Efficient Lyndon Border Array Construction}.
\begin{algorithmic}[1]
\For{$ i = 1  \to n$}
    \State Set $Done[i] = $FALSE
\EndFor
\State Compute $\beta(\s{s})[1\dd n]$
\State Compute SNLB$f$P using Algorithm SNLB$f$P. Say, $\s{s}[1\dd r]$ is the SNLB$f$P.
\For{$ i = n  \to 1$} 
    \If{$Done[i] =$ TRUE}
        \State continue \Comment{i.e., skip what is done below}
    \EndIf
    \State Compute the chain of $\beta(\s{s})[i]$ and push each onto a stack $\mathcal S$.
    \State Compute the penultimate value $k$ of the chain of $\beta(\s{s})[i]$.
    \While{$\mathcal S$ is nonempty}
        \State Pop the value $j$ from the stack
        \If{ $k< r$ }  \Comment{i.e., $\s{s}[1\dd k]$ is a Lyndon word}
              \State $\mathcal \beta(\s{s})[j]=k$
        \Else
              \State $\mathcal \beta(\s{s})[j]=0$
        \EndIf
        \State $Done[j] =$TRUE
    \EndWhile
\EndFor
\end{algorithmic}

Clearly, the time complexity of the above algorithm depends on how many times a chain of borders is traversed. If we can ensure that a chain of borders is never traversed more than once, then the algorithm will surely be linear. To achieve that we use another array $\mathcal Pval[1\dd n]$, initially all set to $-1$. This is required to efficiently compute the penultimate value of a chain of borders. The difficulty here arises because we may need to traverse a part of a chain of borders more than once through different indices because two different indices of the border array, $\beta(\s{s})$, may have the same value. This may incur more cost and make the algorithm super-linear. To avoid traversing any part of a chain of borders more than once we use the array $\mathcal Pval[1\dd n]$ as follows. Clearly, we only need to traverse the chain of borders to compute the penultimate value. So, as soon as we have computed the penultimate value $k$, for a chain, we store the value in the corresponding indices of $\mathcal Pval$. To give an example, suppose we are considering the chain of borders $\beta(\s{s})^{1}[i], \beta(\s{s})^{2}[i], \dd , \beta(\s{s})^{m-1}[i],  \beta(\s{s})^{m}[i]$, where $\beta(\s{s})^{m}[i] = 0$, and suppose that the penultimate value is $k$, i.e., $\beta(\s{s})^{m-1}[i] = k$. Now suppose further that the indices involved in the above chain of borders are $i = i_1, i_2, \dd , i_{m-1}, i_m$. Then as soon as we have got the penultimate value, we update $\mathcal Pval[i_1] = \mathcal Pval[i_2] = \dd = \mathcal Pval[i_{m-1}] = k$. How does this help? If the same chain or part thereof is reached for computing the penultimate value, we can easily return the value from the corresponding $\mathcal Pval[1\dd n]$ entry, and thus we never need to traverse a chain or part thereof more than once. This ensures the linear running time of the algorithm.

\subsection{For Binary Alphabet}
We consider here the case of a binary alphabet $\Sigma$, where $\Sigma = \{a,b\}$ and $a < b$.
Then a binary string $\s{s}$ of length $|\s{s}|=n$, where $\s{s}$ is a non-letter string (i.e., $n > 1$), can be expressed as:

\begin{equation}\label{eq:-binary-rle}
 \s{s}[1 \dd n]=(a)^{p^\prime} (b)^{q\prime} (a)^{p_1} (b)^{q_1} \cdots (a)^{p_m} (b)^{q_m}
\end{equation}

\noindent for $0 \le p^\prime, p_1; \dd; p_m, q^\prime, q_1; \dd; q_m$.
Considering our interest in non-empty binary Lyndon strings, which are
border-free, we require the stronger condition that
$0 < p^\prime, q^\prime < n$
and
$0 \leq p_j, q_j$.
Let $\mathcal{L}$ be the set of binary Lyndon words.

Clearly, Equation. \ref{eq:-binary-rle} represents the \emph{Run Length Encoding} $RLE(\s{s})$ of a binary string $\s{s}$.
For example, if $\s{s} = (a) (b) (a a) (b) (a a a) (b b) (a a) (b) (a a) (b)$, then
$$RLE(\s{s}) = (a)^1 (b)^1 (a)^2 (b)^1 (a)^3 (b)^2 (a)^2 (b)^1 (a)^2 (b)^1$$

We are interested in the values of the parameters (exponents) in Equation. \ref{eq:-binary-rle} ($p^\prime, p_1; \dd; p_m, q^\prime, q_1; \dd; q_m$ ) which we will maintain in an auxiliary linked list $\mathcal R(\s{s})$. 
Note that $\mathcal R[0] =p^\prime $ and $\mathcal R[1] =q^\prime$.
So for the example string $\s{s}$ above, $\mathcal R(\s{s}) = \{ 1,1,2,1,3,2,2,1,2,1 \}$.

\begin{proposition}\label{prop:rle}
For a given string $\s{s} \in \Sigma^+$, the \emph{Run Length Encoding} $RLE(\s{s})$ and subsequently, the list $\mathcal R(\s{s})$ can be computed in time and space linear in the size of the given string $\s{s}$.
\end{proposition}

Now we have the following Proposition to check if a given binary word is a Lyndon word or not (using the \emph{Run Length Encoding} of the binary word).

\begin{proposition}\label{prop:Lyndon-border-binary}
Let $\mathcal{L}$ be the set of Lyndon words over an alphabet $\Sigma$, where $\Sigma = \{a,b\}$ and $a < b$.
For a non-letter word $\s{\ell} = \s{\ell}_1\s{\ell}_2 \cdots \s{\ell}_n$ and its corresponding exponents list $\mathcal R(\mathcal L)$, we have the following:


\begin{description}
 \item [Case $ m = 0$,] we have $p_j, q_j = 0$ for $j \in \{1 \dd m\}$ then
  $\s{\ell} \in \mathcal{L}$ if and only if $p^\prime, q^\prime > 0$.

\item [Case $ m > 0$,] we have $p^\prime, q^\prime, p_j, q_j > 0$  then
   \begin{enumerate}
    \item $\s{\ell} \in \mathcal{L}$ if and only if $p^\prime > p_j$ for $j \in \{1 \dd m\}$.

    \item if $p^\prime = p_m$ and $p^\prime > p_j$ for $j \in \{1 \dd m-1\}$ then $\s{\ell} \in \mathcal{L}$ if and only if $q^\prime < q_m \wedge \lambda < 0$.

    \item if there exists $p_j$ such that $p^\prime =p_j$ and $p^\prime > p_m$ then $\s{\ell} \in \mathcal{L}$ if and only if $q^\prime \le q_j \wedge \lambda < 0$ for $j \in \{1 \dd m-1\}$.

   \end{enumerate}
\end{description}

where $\lambda$ is the index of Lyndon invalid point i.e.
$\mathcal{L}$-fail condition (As defined in
Lemma.~\ref{lemm:Lyndon-invalid-point-primitive} and Lemma.
\ref{lemm:Lyndon-invalid-point-bordered}).

\noindent
\end{proposition}

\begin{lemma}\label{lemm:Lyndon-prefix-binary} For a given string $\s{s}$ of length $n=|\s{s}|$, let the list $\mathcal R$ be the list of exponents of \emph{Run Length Encoding} of $\s{s}$, using the list $\mathcal R$ it can be checked for each proper prefix of $\s{s}$ if it is a Lyndon word or not (computing the array $\Psi(\s{s})$) in $O(n)$ time .
\end{lemma}

\begin{proof}(\textbf{Lemma \ref{lemm:Lyndon-prefix-binary}})
We first focus on computing the $i$-th element $\Psi[i](\s{s})$ of the array $\Psi(\s{s})$.

Clearly, by Proposition. \ref{prop:Lyndon-border-binary}, it takes constant
number of comparisons to determine whether the prefix $\s{s}[1 \dd i]$ is a Lyndon word or not.

Since we have to repeat the steps in
Proposition.\ref{prop:Lyndon-border-binary} for $n$ prefixes, therefore
$\Psi(\s{s})$ can be done in linear time once $\mathcal R(\s{s})$ is computed.

\qed
\end{proof}

While computing the border array $\beta (\s{s})$ (using the Algorithm in \cite{L05}),
we can check conditions (\ref{obs:Lyndon-Word:primitive}), (\ref{obs:Lyndon-Word:unit-factor}) \& (\ref{obs:Lyndon-Word:prefix}) in Observation. \ref{obs:Lyndon-Word} in constant time,
by checking if $\s{\ell}_1 < \s{\ell}_n$ and $ \beta[|\s{\ell}|] \leq 0$.

We can check condition (\ref{obs:Lyndon-Word:suffix}) (Observation. \ref{obs:Lyndon-Word}) by computing the parameters according to Proposition. \ref{prop:Lyndon-border-binary} and checking the $\mathcal{L}$-fail conditions (Lemma.\ref{lemm:Lyndon-invalid-point-primitive} \& \ref{lemm:Lyndon-invalid-point-bordered}).


So, we will evaluate the parameters ($p^\prime, p_1; \dd; p_m, q^\prime, q_1; \dd; q_m$) while computing the lists $\beta[i]$, $\mathcal R$ and $\mathcal L \beta[i]$ simultaneously.

Note that the $i$-th entry of \emph{Lyndon Border Array} $\mathcal L
\beta[i]$ correspond to one of the previous borders (that is also a  Lyndon
word) and hence  we need to consult the value $\Psi[\beta[i]]$, we will set
$\mathcal L\beta[i] = \beta[i]$ if  $\Psi[\beta[i]]=true$. Otherwise, the
algorithm recursively finds the longest border in the chain of borders (at
position $i$) that is also a Lyndon word.

An outline of the algorithm is as follows:
\begin{enumerate}
 \item Compute the border array $\beta(\s{s})$.
 \item Compute the \emph{Run Length Encoding} and subsequently the list $\mathcal R(\s{s})$.
 \item Compute $\Psi(\s{s})$ (Proposition. \ref{prop:Lyndon-border-binary}). 
 \item Check for each prefix of the input string whether or not there exists a border of length $b \not= 0$ such that $\s{s}[1 \dd b] \in \mathcal L$.
\end{enumerate}

A binary Lyndon word can also be expressed in terms of Lyndon properties of the integer parameters (exponents) given by its \emph{Run Length Encoding}. For a binary string $\s{\ell}$, let $\mathcal{R}(\s{\s{\ell}_a})$
 denote the encoding (as a string) of the subsequence of $\s{\ell}$ consisting of all letters $a$ but no letter $b$ ($p^\prime, p_1; \dd; p_m$), similarly $\mathcal{R}(\s{\s{\ell}_b})$ denotes the encoding of the subsequence of $\s{\ell}$ consisting of all letters $b$ ($q\prime, q_1; \dd; q_m$).

\begin{lemma}\label{lemm:Lyndon-encoding-binary}
Let $\s{\ell}$ be a binary word with $\s{\ell}[1]=a$, $\s{\ell}[n]=b$ and associated encodings $\mathcal{R}(\s{\s{\ell}_a})$ and $\mathcal{R}(\s{\s{\ell}_b})$. Then $\s{\ell}$ is a Lyndon word if and only if, either \\
(i) $\mathcal{R}(\s{\s{\ell}_a})$ is a Lyndon word on the alphabet $\{1>2>3> \cdots \}$, or \\
(ii) $\mathcal{R}(\s{\s{\ell}_a})$ is a repetition of a Lyndon word as in (i) and $\mathcal{R}(\s{\s{\ell}_b})$ is a Lyndon word on the alphabet $\{1<2<3< \cdots \}$.

\end{lemma}

\begin{proof}(\textbf{Lemma \ref {lemm:Lyndon-encoding-binary}})
First we consider necessity.
So suppose that (i) holds. Consider any rotation $\s{\s{\ell}_a^r}$ of $\s{\s{\ell}_a}$ (including those with split runs of $a$'s). Then $\mathcal{R}(\s{\s{\ell}_a}) < \mathcal{R}(\s{\s{\ell}_a^r})$ in lexorder over $\{1>2>3> \cdots \}$. Now suppose that (ii) holds. Then for a rotation $\s{\s{\ell}_a^r}$ of $\s{\s{\ell}_a}$, either $\mathcal{R}(\s{\s{\ell}_a}) < \mathcal{R}(\s{\s{\ell}_a^r})$ in lexorder over $\{1>2>3> \cdots \}$, or $\mathcal{R}(\s{\s{\ell}_a}) = \mathcal{R}(\s{\s{\ell}_a^r})$ and $\mathcal{R}(\s{\s{\ell}_b}) < \mathcal{R}(\s{\s{\ell}_b^r})$ in lexorder over $\{1<2<3< \cdots \}$.

For sufficiency, the conditions guarantee that $\s{\ell} \in \mathcal{L}$.

\qed
\end{proof}

\subsubsection{Run time analysis (For Binary Alphabet)}
Computing $\beta(\s{s})$ can be done in linear time Proposition. \ref{prop:border-array}. Also, \emph{run length encoding} (hence the list $\mathcal R$ (the exponents $p^\prime, p_1; \dd; p_m, q^\prime, q_1; \dd; q_m$))
 Proposition. \ref{prop:rle} and the array $\Psi(\s{s})$ can be done in linear time according to Proposition. \ref{prop:Lyndon-border-binary} and Lemma. \ref{lemm:Lyndon-prefix-binary} respectively.

Consequently \emph{Lyndon Border Array} can be computed in linear time once $\beta(\s{s})$, $\mathcal R(\s{s})$ $\Psi(\s{s})$ is computed.

The algorithm requires extra space for keeping the arrays $\Psi(\s{s})$ and $\beta(\s{s})$ (each of length $n$).
Additionally, the algorithm uses an auxiliary extra space to maintain the list $\mathcal R(\s{s})$ of length at most $n$. Therefore the total space required to run the procedure is linear to the length of the given string.

Similar time/space efficiency can be achieved to compute \emph{Lyndon Border Array} by simply applying Duval's algorithm~\cite{Du83} on the exponents of $a$'s and then if necessary on those for $b$'s (the concept presented in Lemma. \ref{lemm:Lyndon-encoding-binary}).

Lemma~\ref{lemm:Lyndon-encoding-binary} shows that a binary string can be decomposed into its Lyndon factors with a double application of Duval's linear algorithm \cite{Du83}, first on the exponents of $a$ and then if necessary on those for $b$.

\section{Lyndon Suffix Array Computation}\label{lba:sec:suffix}
The well-known suffix array of a string records the lexicographically sorted list of all of its suffixes. Our next contribution is to show how the \emph{Lyndon Suffix Array}, like the original suffix array, can be constructed in linear time; for a string of length $n$ it follows that the indexes in the Lyndon variant will be a subset of $\{1, 2, \dots, n\}$. We will exploit the elegant fact that Lyndon suffixes are nested:

\begin{fact}
\label{fact:suff_array}
If the given string \s{s} is a Lyndon word, by Lyndon properties of Lyndon suffixes, the indexes in the Lyndon Suffix Array will necessarily be increasing.
\end{fact}

In order to efficiently construct the Lyndon suffix array we could directly modify the linear-time and space efficient method of Ko and Aluru \cite{KA03} given for the original data structure -- this would involve lex-extension ordering (lexorder for substrings, see \cite{DS14}) along with Fact \ref{fact:suff_array}. We note that the Ko-Aluru method has also recently been adapted to non-lexicographic $V$-order and $V$-letters,
and applied in a novel Burrows-Wheeler transform \cite{DS14}, and hence is quite a versatile technique.
 -- since the modification here will be similar we will just outline the main steps.

Let an {\it $L$-letter} $\s{\ell} = \ell_1 \ell_2 \dd \ell_m$ substring denote the simple case of a Lyndon word such that $\ell_1 < \ell_i$ for $2 \le i \le m$, assumed to be of maximal length, that is, $\ell_1 = \ell_{m+1}$ (if $\ell_{m+1}$ exists); hence $|\s{\ell}| \ge 1$.

Since, apart from the last letter, there may not be Lyndon suffixes of a
string, we perform a linear scan to record the locations of the minimal letter $\ell_1$, say, in the string \s{s}.  Observe also that either an $L$-letter is a Lyndon suffix, or it is the prefix of a Lyndon suffix - the point is that an $L$-letter is a well-defined chunk of text, a substring of the input \s{s}, as opposed to the classic single letter approach. In order to sort chunks of text lexicographically, we will apply \emph{lex-extension} order defined as follows.

\begin{definition} Suppose that according to some factorization $\mathcal F$, two strings
$\s{u, v} \in \Sigma^+$ are expressed in terms of nonempty factors:
$\s{u} = u_1 u_2 \cdots u_m, \s{v} = v_1 v_2 \cdots v_n$.
Then $\s{u} <_{LEX(F)} \s{v}$ if and only if one of the following holds:\\
(1) \s{u} is a proper prefix of \s{v} (that is, $u_i = v_i$ for $1 \le i \le m < n)$; or\\
(2) for some $i \in 1\dd min(m, n), u_j = v_j$ for $j = 1, 2, \dd , i-1$, and $u_i < v_i$ (in lexicographic order).
\end{definition}

 First, using a linear scan we apply the $\mathcal L \beta$ to record the indexes of all the Lyndon suffixes of $\s{s}$.
The factorization $\mathcal F$ which we will use is that of decomposing the input into substrings of $L$-letters; in the case of unit length  $L$-letters we concatenate them so that $\mathcal F$ has the general form:
$$ \s{u} \ell_{1}^{i_1} \s{\ell_1} \ell_{1}^{i_2} \s{\ell_2} \cdots \ell_{1}^{i_t} \s{\ell_t}$$
where $\s{u} \in \Sigma^*$ does not contain $\ell_{1}$, $t \ge 1$, each $i_j \ge 0$, every $\s{\ell_j}$ is an $L$-letter and $|\s{\ell_t}| \ge 1$.
-- in practice this just entails keeping track of occurrences of the letters $\s{\ell}_1$ in \s{s}.
Since we are computing Lyndon suffixes we can ignore $\s{u}$ and hence assume that $\mathcal F$ has the form
$\ell_{1}^{i_1} \s{\ell_1} \ell_{1}^{i_2} \s{\ell_2} \cdots \ell_{1}^{i_t} \s{\ell_t}$.

An $\ell$-suffix is a suffix of $\mathcal F$ commencing with $\ell_{1}^{i_j}$. We now apply the Ko-Aluru linear method, consisting of three main steps, to the $\ell$-suffixes:
which we first outline followed by further detail for each.

\begin{itemize}
\item Using a linear scan of the input string $\s{s}$ and lex-extension
ordering, divide all $\ell$-suffixes of \s{s} into two types: 
That is, let $\s{s_i}$ denote the suffix starting at index $i$, so
$\s{s_i}=\s{s}[i \dd n]$. Then the type $S$ Lyndon suffixes are the set 
$\{\s{s_i} | \s{s_i} < \s{s_{i+1}}\}$ and the type $L$ Lyndon sufixes are the
set $\{\s{s_j} | \s{s_j} > \s{s_{j+1}}\}$.

\item Sort by lex-extension  all $\ell$-suffixes of type $S$ in $O(n)$-time
using a modified Bucket Sort followed by recursion on at most half of the
string.

\item Using a linear scan obtain the lex-extension order of all remaining
$\ell$-suffixes (assumed to be type $L$) from the sorted ones. This step is
obtained from observing that the type $L$ $\ell$-suffixes occurring in \s{s}
between two type $S$ $\ell$-suffixes, $S_i$ and $S_j$ where $i<j$, are already
ordered such that $L_l <_{LEX(\mathcal F)} L_k$ for $i<k<l<j$.

\end{itemize}

At this stage we have computed an $\ell$-suffix array. There are two final steps for computing the Lyndon suffix array. Firstly, the classic suffix array for the last $L$-letter $\s{\ell_t}$ (without the prefix $\ell_1$) is processed directly using the Ko-Aluru method and the indexes inserted into the $\ell$-suffix array. Then we perform a linear scan of the $\ell$-suffix array, and applying Fact \ref{fact:suff_array}, by selecting only the sequence of increasing integers yields the Lyndon suffix array.

The last suffix is both type $S$ and $L$.
The modification from lexicographic to lex-extension order follows from, firstly the linear factorization of the input into $L$-letters, and secondly that lex-extension ordering applies lexicographic order pairwise to $L$-letter substrings which each requires no more than time linear in the length of the $L$-letters -- hence $O(n)$ overall. The space efficiency follows from the original method \cite{KA03}.

\subsection{A Simpler algorithm for computing a Lyndon Suffix Array from a Suffix Array}\label{lba:subsec:simpleLSA}
We present an alternative simple algorithm, derived from the classic suffix array, which also exploits the nested structure expressed in Fact \ref{fact:suff_array}. Suppose we are given the suffix array of the string $\s{s}$. Now our algorithm finds the largest suffix (max), and then searches inside it to find the
second largest one and so on, taking advantage of the fact that the suffixes are already sorted in the suffix array.
Repeatedly finding the max value can be implemented efficiently using the Range Minimum Query (RMQ)  \cite{DBLP:conf/latin/BenderF00}, which requires $O(n)$ time pre-processing and then $O(1)$ time for each query.

\begin{lemma}{\textbf{Maximum range suffixes are Lyndon words:}}\label{lemm:Lyndon-RMQ}
Suppose we are given the suffix array $\mathcal A$ of a string $\s{s}$. Let the set $\mathcal M$ be the set of maximum values of the range of suffixes $\mathcal A[0\dd i]$, where $i$ is the index of the $i$-th suffix $\ell_i$.
Each suffix $\ell \in \mathcal M$ is a Lyndon word.
\end{lemma}

\begin{proof}[Lemma~\ref{lemm:Lyndon-RMQ}]
Suppose $\ell$ is a max range suffix ($\ell \in \mathcal M$) at index $i$ with order value $SA[i] = m$.
  Suppose there exists a suffix $\s{\ell}^{\prime\prime}$ of $\s{\ell}$ at index $i^{\prime\prime}$ (w.r.t. to $\mathcal A$) with order value $\mathcal A[i^{\prime\prime}] = m^{\prime\prime}$, and also suppose $\s{\ell}^{\prime\prime} < \s{\ell}$. Hence $m^{\prime\prime} > m$ ($\mathcal A[i^{\prime\prime}]> \mathcal A[i]$) where $i^{\prime\prime} < i$, which contradicts the fact that $m$ is the maximum value in the range $\mathcal A[0\dd i]$. Therefore, $\s{\ell}$ is strictly smaller than all of its proper suffixes --
we conclude that $\ell$ is a Lyndon word.

\qed
\end{proof}

The linear-time method for computing the Lyndon suffix array is very simple. Below, we first outline the steps followed by the pseudo-code.


\begin{enumerate}
  \item Compute the suffix array $\mathcal A$ of $\s{s}\$$.
  \item Find the value $max = \text{MAX}(\mathcal A[0\dd n])$ and its index $i$ in the suffix array, then add the value $max$ to Lyndon Suffix array $\mathcal LS$.
  \item Find the value $max$ and its index $i$ in the range $\mathcal A[0\dd i]$, $max = \text{MAX}(\mathcal A[0 \dd i])$ and $i= \text{Indexof}(max, \mathcal A)$, then add the value $max$ to $\mathcal LS$.
  \item Repeat step 2 until the set $\mathcal A[0\dd i]$ is empty.
\end{enumerate}

\begin{algorithm}[H]
\caption{Computing the Lyndon Suffix Array from a Suffix Array.}\label{algo:simple-LSA}
\begin{algorithmic}
\Procedure{ComputeLSA}{$\s{s}$}
\State $n \leftarrow |s|$; \ $\mathcal L[1\dd n] \leftarrow (-1)^n$	
\State $\vartriangleright$\text{ compute the suffix array of \s{s}\$}
\State $\mathcal A \leftarrow SA(\s{s}\$)$
\State $i\leftarrow n$; $j\leftarrow n$; $max\leftarrow 0$
\While{$(\mathcal A[0\dd i]  \not= \emptyset)$}
\State $\vartriangleright$\text{ find the maximum value (and its index) in the range $\mathcal A [0\dd i]$.}
\State $(max, i) \leftarrow (\mathcal A[idx], idx) \ | \ idx = \operatorname{arg\,max}_{idx} (A[idx]) \text{ for }  0 \leq idx < i $
\State $j \leftarrow j-1$
\State $\mathcal L[j] \leftarrow max$
\EndWhile
\State return $\mathcal L$
\EndProcedure
\end{algorithmic}
\end{algorithm}


\section{Conclusion}\label{lba:sec:conclusion}
In this article, we have extended two well-known data structures in
stringology. We first have adapted the concept of a border array to
introduce the \itbf{Lyndon Border Array} $\mathcal L \beta$ of a string
$\s{s}$, and have described a linear-time and linear-space algorithm for
computing $\mathcal L \beta (\s{s})$. Furthermore, we have defined the
\itbf{Lyndon Suffix Array}, which is an adaptation of the classic suffix
array. By modifying the linear-time construction of Ko and Aluru
\cite{KA03} we similarly achieve a linear construction for our Lyndon
variant. We also present a simpler algorithm to construct a \itbf{Lyndon
Suffix Array} from a given Suffix Array.

The potential value of the Lyndon Border Array is that it allows for deeper
burrowing into a string to yield paired Lyndon patterned substrings. The
Lyndon suffix array lends itself naturally to searching for Lyndon patterns
in a string.
If the given text or string has a sparse number of Lyndon words (as likely
in English literature due to the vowels $a,e$ often occurring internally in
words), then the Lyndon suffix array may offer efficiencies. Polyrhythms,
or cross-rhythms,
are when two or more independent rhythms play at the same time -- nested
Lyndon suffixes can exist in these rhythms. We propose that applications of
these specialized data structures might arise in the context of the
relationship existing between de Bruijn sequences and Lyndon words
\cite{Necklaces78}.

We propose that applications of \itbf{Lyndon Border Array} may arise in
combinatorics  in relation to the Christoffel words.
The methods used for these problems often make use of structures equivalent
to suffix trees in order to achieve efficient execution.

\bibliographystyle{alpha}
\bibliography{ref}

\end{document}